\documentclass[submission,copyright,creativecommons,sharealike]{eptcs}
\providecommand{\event}{HOR 2010} 
\usepackage{breakurl}             
\usepackage[usenames,dvipsnames]{color}
\usepackage{amsmath}
\usepackage{amssymb}
\usepackage{bm}
\usepackage{theorem}
\usepackage{enumerate}
\usepackage{prooftree}

\newcommand{\newsection}[1]{
\onlyEa{\vspace{-0.4cm}}
\section{#1}
\onlyEa{\vspace{-0.3cm}}
}
\newcommand{\newsubsection}[1]{
\onlyEa{\vspace{-0.3cm}}
\subsection{#1}
\onlyEa{\vspace{-0.2cm}}
}
\newcommand{\comment}[1]{}

\newcommand{\red}[1]{\ensuremath{\underset{#1}{\rightarrow}}}

\newcommand{\hred}{\red{h}}

\newcommand{\hRedP}[1]{\ensuremath{\underset{#1}{\rightsquigarrow}}}

\newcommand{\dev}{\ensuremath{\triangleright}}
\newcommand{\sdev}{\ensuremath{\blacktriangleright}}
\newcommand{\devn}{\ensuremath{\dev^*}}

\newcommand{\idev}{\ensuremath{\stackrel{int}{\triangleright}}}

\newcommand{\idevn}{\ensuremath{\stackrel{int}{\triangleright^*}}}
\newcommand{\hredn}{\ensuremath{\hred^*}}
\newcommand{\hRedPn}[1]{\ensuremath{\underset{#1}{\overset{*}{\rightsquigarrow}}}}
\newcommand{\sepdev}{\ensuremath{\underset{h}{\triangleright}}}
\newcommand{\subsepdev}{\ensuremath{\underset{h}{\blacktriangleright}}}

\renewcommand{\l}{\lambda}
\newcommand{\fv}[1]{{\tt fv}(#1)}
\newcommand{\var}[1]{{\tt var}(#1)}
 
\newcommand{\appfrom}[1]{#1_1 #1_2}

\newcommand{\dom}{{\tt dom}}
\newcommand{\set}[1]{\{ #1 \}}

\newcommand{\match}[3]{\ensuremath{{#1}\ll^{#3}{#2}}}
\newcommand{\negMatch}[2]{\ensuremath{{#1}\not\ll{#2}}}

\newcommand{\metaexists}{\ensuremath{\pmb{\mathsf{\exists}}}}
\newcommand{\metaforall}{\ensuremath{\pmb{\forall}}}
\newcommand{\sthat}{\ensuremath{\textnormal{ s.t. }}}
\newcommand{\setsthat}{\ensuremath{\ : \ }}
\newcommand{\metadot}{\ensuremath{\ . \ }}

\newcommand{\carlos}[1]{{\color{red} Carlos: #1}}

\newcommand{\sep}{\hspace{.5cm}}

\newcommand{\schemeName}[1]{\ensuremath{\mathsf{#1}}}
\newcommand{\setName}[1]{\ensuremath{\textbf{#1}}}
\newcommand{\regla}[3]{
    \prooftree
         #1
    \justifies  
         #2
    \thickness=0.05em
    \using
        #3
    \endprooftree}
\newcommand{\reglaangosta}[3]{
    \prooftree
         #1
    \justifies  
         #2
    \thickness=0.05em
    \using
        #3
    \endprooftree}
\newcommand{\reglaalta}[3]{
    \prooftree
         #1 \vspace{1mm}
    \justifies  
         #2
    \thickness=0.05em
    \using
        #3
    \endprooftree}

\newcommand{\textand}{\textnormal{ and }}

\newcommand{\parteuno}{\textit{(i)}}
\newcommand{\partedos}{\textit{(ii)}}
\newcommand{\partetres}{\textit{(iii)}}
\newcommand{\abajito}{\hspace{1mm} \\ \vspace{-4mm}}
\newcommand{\abajo}{\hspace{1mm} \\}

\newtheorem{theorem}{Theorem}[section]
\newtheorem{lemma}[theorem]{Lemma}

\newtheorem{corollary}[theorem]{Corollary}
\newenvironment{proof}[1][Proof]{\begin{trivlist}
\item[\hskip \labelsep {\bfseries #1}]}{\hspace{\stretch{1}} $\Box$ \end{trivlist}}
\newtheorem{definition}[theorem]{Definition}

\newcommand{\reflem}[1]{L.~\ref{lem:#1}}
\newcommand{\refeq}[1]{(\ref{eq:#1})}

\newcommand{\qed}{\nobreak \ifvmode \relax \else
      \ifdim\lastskip<1.5em \hskip-\lastskip
      \hskip1.5em plus0em minus0.5em \fi \nobreak
      \vrule height0.75em width0.5em depth0.25em\fi}

\newcommand{\DRefl}{\schemeName{DRefl}}
\newcommand{\DAbs}{\schemeName{DAbs}}

\newcommand{\DApp}{\schemeName{DApp}} 
\newcommand{\DBeta}{\schemeName{DBeta}}

\newcommand{\HAppL}{\schemeName{HApp1}}
\newcommand{\HBeta}{\schemeName{HBeta}}
\newcommand{\HPat}{\schemeName{HPat}}
\newcommand{\PatHead}{\schemeName{PatHead}}
\newcommand{\PatL}{\schemeName{Pat1}}
\newcommand{\PatR}{\schemeName{Pat2}}

\newcommand{\IRefl}{\schemeName{IRefl}}
\newcommand{\IAbs}{\schemeName{IAbs}}

\newcommand{\IAppL}{\schemeName{IApp1}} 
\newcommand{\IAppR}{\schemeName{IApp2}}

\newcommand{\PMatch}{\schemeName{PMatch}}
\newcommand{\PConst}{\schemeName{PConst}}
\newcommand{\PNoCData}{\schemeName{PNoCData}}
\newcommand{\PCDataNoL}{\schemeName{PCDataNo1}}
\newcommand{\PCDataNoR}{\schemeName{PCDataNo2}}
\newcommand{\PCDataNoT}{\schemeName{PCDataNo3}}

\newcommand{\onlyEa}[1]{}
\newcommand{\onlyFull}[1]{#1}
\newcommand{\diffEaFull}[2]{#2}

\title{A standardisation proof for algebraic pattern calculi}
\author{Delia Kesner
\institute{PPS, CNRS and Universit\'e Paris Diderot \\
France}
\email{Delia.Kesner@pps.jussieu.fr}
\and
Carlos Lombardi
\institute{Depto. de Ciencia y Tecnolog\'ia \\ 
Univ. Nacional de Quilmes \\ Argentina}
\email{clombardi@unq.edu.ar}
\and
Alejandro R\'ios
\institute{Depto. de Computaci\'on \\ 
Facultad de Cs. Exactas y Naturales \\
Univ. de Buenos Aires -- Argentina}
\email{rios@dc.uba.ar}
}
\def\titlerunning{Standardisation for constructor based pattern calculi}
\def\authorrunning{D.Kesner, C.Lombardi \& A.R\'ios}

\begin{document}
\maketitle

\vspace{-4mm}
\begin{abstract}

This work gives some insights and results on 
standardisation for call-by-name pattern calculi.
More precisely, we  define  standard reductions for a pattern calculus with
constructor-based data terms and patterns. This notion is based on
reduction steps that are needed to match an argument with respect to a
given pattern.
We  prove  the  Standardisation Theorem  by  using the  technique
developed by  Takahashi~\cite{takahashi-std} and Crary~\cite{crary-std}  for $\l$-calculus.  
The  proof is
based on the fact that any  development can be specified as a sequence
of  head steps  followed by  internal  reductions, i.e.
reductions  in which  no head  steps are
involved. 

\comment{\carlos{Analyse the differences between the
  proofs by Takahashi and by Barendregt, as Barendregt seems to
  define an internal development concept -- result of a 5-minute
  reading}}

\end{abstract}

\vspace{-2mm}

\newsection{Introduction}

{\bf  Pattern   Calculi}:  Several  calculi,   called  \textit{pattern
calculi}, have been proposed in  order to give a formal description of
pattern matching;  i.e. the ability to analyse  the different possible
forms of the argument of a function in order to decide among different
alternative definition clauses.

The  \textit{pattern   matching}  operation  is  the   kernel  of  the
evaluation  mechanism  of   all  these  formalisms,  basically  because
reduction  can only  be  fired when  the  argument passed  to a  given
function  matches its  pattern specification.  An analysis  of various
pattern  calculi  based  on  different  notions  of  pattern  matching
operations  and  different  sets  of  allowed patterns  can  be  found
in~\cite{jk-fcp}.

\onlyFull{\medskip}
\noindent {\bf Standardisation}: A fundamental result in the $\l$-calculus is the
\textit{Standardisation Theorem}, which states that if a term $M$
$\beta$-reduces  to  a  term  $N$,  then there  is  a  {\it  standard}
$\beta$-reduction  sequence from $M$  to $N$  which can  be seen  as a
canonical way to reduce  terms.  This result has several applications,
e.g. it is used to  prove the non-existence of reduction between given
terms.  One  of its  main corollaries is  the quasi-leftmost-reduction
theorem, which  in turn is used  to prove the  non-existence of a normal
form for a given term.

A first study on
standardisation for call-by-name $\l$-calculus
appears in~\cite{curry-feys}. Subsequently, several standardisation methods 
have been devised, for
example~\cite{barendregt} Section 11.4,~\cite{takahashi-std},
~\cite{kashima-std} and~\cite{plotkin-std}. 

While  leftmost-outermost  reduction  gives  a standard  strategy  for
call-by-name $\l$-calculus, more  refined notions of reductions  are necessary to
define        standard       strategies        for       call-by-value
$\l$-calculus \cite{plotkin-std},     first-order     term    rewriting
systems~\cite{HL91,Terese03},    Proof-Nets~\cite{girard},   etc.   

All standard reduction strategies  require the definition of some {\it
selected} redex by means of  a partial function from terms to redexes;
they  all give  priority  to  the selected  step,  if possible.   This
selected  redex is  sometimes called  {\it external}~\cite{MelliesTh},
but we will refer here to it as the \textit{head redex} of a term.

It is also worth mentioning a generic standardisation proof~\cite{PDR04}
that can uniformly treat cal-by-name and call-by-value
$\l$-calculus. It is parameterized over the set of values that allow
to fire the beta-reduction rule. However, the set of values are defined
there in a global sense, while in pattern calculi 
being a value strongly depends on the form of the given pattern.

\onlyFull{\medskip}
 
\noindent {\bf Standardisation in Pattern  Calculi}: For call-by-name $\l$-calculus, any term of the form
$(\l x.M)N$  is a redex,  and the  head redex for  such a term  is the
whole term.  In pattern calculi any  term of the form $(\l p.M)N$ is a
redex candidate,  but not  necessarily a redex.  The parameter  $p$ in
such terms can  be more complex than a single  variable, and the whole
term is not a redex
if the argument $N$ does not match $p$, i.e., if $N$ does not verify 
the structural conditions imposed by $p$.  
In this case
we will choose
as head a reduction step lying inside $N$ (or even inside
$p$) which makes $p$ and $N$ be closer to a possible match.
While this situation bears some resemblance with 
\textit{call-by-value} $\l$-calculus~\cite{plotkin-std}, there is an
important difference: both the fact of $(\l p.M) N$ being a redex, and
whether a redex inside $N$ could be useful to get $p$ and $N$
closer to a possible match, depend on \textit{both} $N$ \textit{and} $p$.

The aim of this contribution
is  to analyse  the  existence of  a standardisation procedure for
pattern calculi  in a direct way, i.e. without using any complicated encoding of such
calculi into  some general computational  framework~\cite{kvodv}. This
direct  approach aims to put to  evidence  the  fine  interaction  between
reduction      and     pattern      matching,     and      gives     a
standardisation \textit{algorithm}  which is  specified in terms  of the
combination  of   computations  of  independent   terms  with  partial
computations of terms  depending on some pattern.  We  hope to be able
to  extend this  algorithmic  approach to  more sophisticated  pattern
calculi handling open  and dynamic patterns~\cite{JK05}.

\onlyEa{
The paper is organized as follows. Section \ref{s:calculus} introduces
the calculus, Sections~\ref{s:concepts} and ~\ref{s:lemmas} give, respectively, the
main concepts and ideas needed for the standardisation proof and the main results, and Section~\ref{s:conclusion}
concludes and gives future research directions.
}

\onlyFull{
\medskip
The paper is organized as follows. Section \ref{s:calculus} introduces
the calculus, Section~\ref{s:concepts} gives the main concepts needed for the standardisation proof and the main results, Section~\ref{s:prelimResults} presents some lemmas used in the main proofs, Sections~\ref{s:separability} and~\ref{s:standardisation} show the main results used in the Standardisation Theorem proof and then the theorem itself; finally, Section~\ref{s:conclusion} concludes and gives future research directions.
}

\newsection{The calculus}
\label{s:calculus}

We will study a very simple form of pattern calculus, consisting of
the extension of standard $\l$-calculus with a set of constructors
and allowing constructed patterns. This calculus appears for
example in Section 4.1 in~\cite{jk-fcp}.

\begin{definition}[Syntax]
The calculus is built upon two different enumerable sets of symbols, the \textit{variables} $x,y,z,w$ and the \textit{constants} $c,a,b$; its syntactical categories are:
$$
\begin{array}{lrcllrcl}
{\bf Terms}        & M,N,Q,R & :: = & x \mid c \mid \l p. M \mid MM  &
{\bf DataTerms}    & D &       ::  = & c \mid DM \\
{\bf Patterns}     & p,q &     :: = & x \mid d &
{\bf DataPatterns} & d &       :: = & c \mid d p \\
\end{array}
$$
\end{definition}

Free and bound variables of terms are defined as expected as well as $\alpha$-conversion.

\onlyFull{
\begin{definition}[Substitution]
A susbsitution $\theta$ is a function from variables to terms with finite
domain, where $\dom(\theta)= \{ x \setsthat \theta(x) \neq x \}$. The extension of $\theta$ to terms is defined as expected. We denote $\theta \ ::= \ \set{x_1/M_1, \ldots, x_n/M_n}$ wherever $\dom(\theta) \subseteq \{ x_1, \ldots, x_n\}$.
Moreover, for $\theta, \nu$ substitutions, $X$ a set of variables, we define
$$
\begin{array}{rcl}
\var{\theta}  & :: = & \dom(\theta) \,\bigcup\; 
	                     \big(\cup_{x \in \dom(\theta)} \fv{\theta x}\big)  \\
\nu \theta    & :: = & \big(\cup_{x \in \dom(\theta)} \set{x / \nu(\theta x) }\big) 
	                     \,\bigcup\; 
	                     \big(\cup_{x \in (\dom(\nu) -\dom(\theta))} \set{x / \nu x }\big)  \\
\theta \mid_X & :: = & \cup_{x \in X \cap \dom(\theta)} \set{x / \theta x }    
\end{array}
$$
\end{definition}
}

\begin{definition}[Matching]
Let $p$ be a pattern and $M$ a term which do not share common variables.
Matching on $p$ and $M$ is a partial function yielding a substitution and defined by the following rules ($\uplus$ on substitutions denotes disjoint union with respect to their domains, being undefined if the domains have a non-empty intersection):

\begin{center}
$
\begin{array}{c}
\regla
    {}
    {\match{x}{N}{\set{x/N}}}
    {}
\hspace{1cm}
\regla
    {}
    {\match{c}{c}{\emptyset}}
    {}
\hspace{1cm}
\regla
    {\match{d}{D}{\theta_1} 
        \sep \match{p}{N}{\theta_2} 
        \sep {\theta_1 \uplus \theta_2} \textnormal{ defined}}
    {\match{dp}{DN}{\theta_1 \uplus \theta_2}}
    {}
\end{array}$
\end{center}

We write $\match{p}{M}{}$ iff $\exists \theta\ \match{p}{M}{\theta}$. 
Remark that $\match{p}{M}{}$ implies that $p$ is linear.
\end{definition}

\begin{definition}[Reduction step]
We consider the following reduction steps modulo $\alpha$-conversion:
\begin{center}
$
\begin{array}{cccc}
\reglaangosta
    {M\red{} M'}
    {M\,N\red{} M'\,N}
    {\schemeName{SAppL}} & 
\reglaangosta
    {N\red{} N'}
    {M\,N\red{} M\,N'}
    {\schemeName{SAppR}} & 
\reglaangosta
    {\match{p}{N}{\theta}}
    {(\l p.M)\,N\red{} \theta M}
    {\schemeName{SBeta}} & 
\reglaangosta
    {M\red{} M'}
    {\l p.M\red{}\l p.M'}
    {\schemeName{SAbs}}
\end{array}
$
\end{center}
\onlyEa{\vspace{-2mm}}
\end{definition}

By working modulo
$\alpha$-conversion  we can always assume
in rule (\schemeName{SBeta}) that
$p$ and $N$ do not share
common variables in order to compute matching.

\onlyFull{
\begin{lemma}[Basic facts about the calculus] \label{lem:calculus-basics} \abajito
\begin{enumerate}[a.]
	\item \label{bfc:b+c}  
             (data pattern/term structure)  Let $d \in \setName{DataPatterns}$ (resp. $D \in \setName{DataTerms}$), 
             then $d = c p_1 \ldots p_n$ (resp. $D = c M_1 \ldots M_n$) for some $n \geq0$. 
	\item \label{bfc:d}  
              (data patterns only match data terms) Let $d \in \setName{DataPatterns}$, $M$ a term, such that $\match{d}{M}{}$. Then $M \in \setName{DataTerms}$.
	\item \label{bfc:e}  
              (minimal matches) If $\match{p}{M}{\theta}$ then $\dom(\theta) = \fv{p}$.
	\item \label{bfc:f} 
              (uniqueness of match) If $\match{p}{M}{\theta_1}$ and $\match{p}{M}{\theta_2}$, then $\theta_1 = \theta_2$.
\end{enumerate}
\end{lemma}
}

Crucial to the standardisation proof is the concept of development, we formalize it through the relation $\dev$ 
, meaning
$M \dev N$ iff there is a development (not necessarily complete) with source $M$ and target $N$.
\vspace{-5mm}

\begin{definition}[Term and substitution development]
We define the relation $\dev$ on terms and a corresponding relation $\sdev$ on substitutions.
The relation $\dev$ is defined by the following rules:
\begin{center}
$
\begin{array}{c}
\reglaangosta
    {}
    {M\dev M}
    {\DRefl}
\hspace{1cm}
\reglaangosta
    {M\dev M'}
    {\l p.M\dev \l p.M'}
    {\DAbs}
\diffEaFull{\hspace{1cm}}{\\ \\}
\reglaangosta
    {M\dev M'
    \quad
     N\dev N'}
    {M\,N\dev M'\,N'}
    {\DApp} 
\hspace{1cm}
\reglaangosta
    {M\dev M'
    \quad
    \theta \sdev \theta'
    \quad 
     \match{p}{N}{\theta}}
    {(\l p.M)\,N\dev \theta' M'}
    {\DBeta}
\\
\end{array}
$ \\
\end{center}

\noindent and $\sdev$ is defined as follows: $\theta \sdev \theta'$ iff
$\dom(\theta) = \dom(\theta')$ and $\forall x\in \dom(\theta) \metadot \theta x \dev \theta' x$
\end{definition}

\onlyFull{\newsubsection{Head step}}
The definition of head step will take into account the terms $(\l p.M) N$ 
even if $\negMatch{p}{N}$. In such cases, the head redex will be inside $N$ as the patterns in this calculus are always normal forms (this will not be the case for more complex pattern calculi). 

The selection of the head redex inside $N$ depends on both $N$ and
$p$.  This differs from standard call-by-value $\l$-calculus, where
the selection depends only on $N$.

We show this phenomenon with a simple example. Let $a,b,c$ be
constants and  $N = (a R_1) R_2$, where $R_1$ and $R_2$ are
redexes. The redexes in $N$ needed to achieve a match with a
certain pattern $p$, and thus the selection of the head redex, depend on
the pattern $p$.

Take for example different patterns $p_1 = (a x) (b y), p_2 = (a b x)
y, p_3 = (a b x) (c y), p_4 = (a x) y$, and consider the term 
$Q = (\l p.M) N$. 
If $p = p_1$, then it is not necessary to reduce $R_1$
(because it already matches $x$) but it is necessary to reduce
$R_2$, because no redex can match the pattern $b y$; hence $R_2$ will
be the head redex in this case. Analogously, for $p_2$ it is necessary
to reduce $R_1$ but not $R_2$, for $p_3$ both are needed (in this case
we will choose the leftmost one) and $p_4$ does match $N$, hence the
whole $Q$ is the head redex.
This observation motivates the following definition.

\begin{definition}[Head step] \label{dfn:hred}
The relations $\hred$ (head step) and $\hRedP{p}$ (preferred needed step to match pattern $p$) are defined as follows:
\begin{center}
$
\begin{array}{c}
\reglaalta
    {M\hred M'}
    {M\,N\hred M'\,N}
    {\HAppL}
\hspace{1cm}
\regla
    {\match{p}{N}{\theta}}
    {(\l p.M)\,N\hred \theta M}
    {\HBeta}
\hspace{1cm}
\reglaalta
    {N\hRedP{p} N'}
    {(\l p.M)\,N\hred (\l p.M)\,N'}
    {\HPat}
\\
\\
\reglaalta
    {M\hred M'}
    {M\hRedP{d}M'}
    {\PatHead}
\hspace{1cm}
\reglaalta
    {D \hRedP{d} D'}
    {D M \hRedP{dp} D' M}
    {\PatL}
\hspace{1cm}
\reglaalta
    {M \hRedP{p} M'
     \quad
     \match{d}{D}{}}
    {D M \hRedP{dp} D M'}
    {\PatR}
\end{array}
$
\end{center}
\end{definition}
The rule \PatHead\ is intended for data patterns only, not being valid for variable patterns; we point this by writing a $d$ (data pattern) instead of a $p$ (any pattern) in the arrow subscript inside the conclusion.

We observe that the rule analogous to \HPat\ in the presentation of
standard reduction sequences for call-by-value
$\l$-calculus in both \cite{plotkin-std} and \cite{crary-std} reads
\vspace{-1mm}
$$ \frac{N \hred N'}{(\l p.M)N \hred (\l p.M)N'} $$ 
\vspace{-3mm} \\
reflecting the $N$-only-dependency feature aforementioned.

We see also that a head step in a term like $(\l p.M) N$ determined by rule \HPat\ will lie inside $N$, 
but the same step will not necessarily be considered head if we analyse $N$ alone. 

It is easy to check that if $M \hRedP{p} M'$ then $\negMatch{p}{M}$, avoiding any overlap between \HBeta\ and \HPat\, and also between \PatL\ and \PatR. 
This in turn implies that all terms have at most one head redex.
We remark also that the head step depends not only on the pattern structure but also on the match or lack of match between pattern and argument.

\onlyFull{
\begin{lemma}[Basic facts about head steps] \label{lem:head-step-basics} \abajito
\begin{enumerate}[a.]
	\item \label{bf:a} (head reduction only if abstraction in head) 
			Let $M$ be a term such that $M \hred M'$ for some $M'$. 
			Then $M = (\l p.M_{01}) M_1 \ldots M_n$ with $n \geq 1$.
	\item \label{bf:b} (head reduction only if no match) 
			Let $M$ be a term such that $M \hred M'$ for some $M'$, $d \in \setName{DataPatterns}$. 
			Then $\negMatch{d}{M}$.
	\item \label{bf:c} ($\hRedP{p}$ only if $\hred$ or data term) 
			Let $p$ be a pattern and let $M$ be a term such that $M \hRedP{p} M'$ for some $M'$. 
			Then either $M \in \setName{DataTerms}$ or $M \hred M'$.
\end{enumerate}
\end{lemma}

\begin{proof}
Item~(\ref{bf:a}) is trivial.
Item~(\ref{bf:b}) uses Item~(\ref{bf:a})
and \reflem{calculus-basics}:(\ref{bfc:d}). 
Item~(\ref{bf:c}) is trival by definition of $\hRedP{p}$.
\end{proof}
}

\newsection{Main concepts and ideas needed for the standardisation proof}
\label{s:concepts}
In order to build a standardisation proof for constructor based
pattern calculi we chose to adapt the one in~\cite{takahashi-std} for
the call-by-name $\l$-calculus, later adapted to call-by-value
$\l$-calculus in~\cite{crary-std}, over the classical presentation
of \cite{plotkin-std}.

The proof method relies on a \textbf{h-development} property
stating that any
development can be split into a leading sequence of head steps followed by a
development in which no head steps are performed; this is our
Corollary~\ref{lem:crary-6} which corresponds to the so-called ``main lemma''
in the presentations by Takahashi and Crary.

Even for a simple form of pattern calculus such as the one presented
in this contribution, both the definitions (as we already mentioned
when defining head steps) and the proofs are non-trivial 
extensions of the 
corresponding ones for standard $\l$-calculus, even in the
framework of call-by-value.
As mentioned before, the reason is the need to take into
account, for terms involving the application of a function to an
argument, the pattern of the function parameter when deciding
whether a redex inside the argument should be considered as a
head redex.

In order to formalize the notion of ``development without occurrences of head steps'', an \textit{internal development} relation will be defined.
The dependency on both $N$ and $p$ when analysing the reduction steps from a term like $(\l p.M)N$ is shown in the rule \IAppR.

\begin{definition}[Internal development] \label{dfn:idev}
The relations $\idev$ (internal development) and $\idev_p$ (internal development with respect to the pattern $p$) are defined as follows:
\begin{center}
$
\begin{array}{c}
\regla
    {}
    {M\idev M}
    {\IRefl}
\hspace{1cm}
\regla
    {M\dev M'}
    {\l p.M\idev \l p.M'}
    {\IAbs}
\hspace{1cm}
\regla
    {M\neq \l p.M_1
    \quad
     M\idev M'
    \quad
     N\dev N'}
    {M\,N\idev M'\,N'}
    {\IAppL} 
\diffEaFull{\vspace{3mm}}{\\}
\\
\regla
    {M \dev M'
    \quad
     N \idev_p N'}
    {(\l p.M)\,N \idev (\l p.M') N'}
    {\IAppR}
\hspace{1cm}
\regla
    {N \dev N' \quad \match{p}{N}{}}
    {N \idev_p N'}
    {\PMatch}
\hspace{1cm}
\regla
    {N \idev N'}
    {N \idev_{c} N'}
    {\PConst}
\diffEaFull{\vspace{3mm}}{\\}
\\
\regla
    {N \notin \setName{DataTerms}
    \quad
    N \idev N'}
    {N \idev_{dp} N'}
    {\PNoCData}
\hspace{1cm}
\regla
    {D \idev_{d} D'
    \quad
     M \dev M'
    \quad
    \negMatch{d}{D} }
    {D M \idev_{dp} D' M'}
    {\PCDataNoL}
\diffEaFull{\vspace{3mm}}{\\}
\\
\regla
    {D \dev D'
    \quad
     M \idev_p M'
    \quad
    \match{d}{D}{}
    \quad
    \negMatch{p}{M} }
    {D M \idev_{dp} D' M'}
    {\PCDataNoR} 
\diffEaFull{\vspace{3mm}}{\\}
\\
\regla
    {D \dev D'
    \quad
     M \dev M'
    \quad
    \match{d}{D}{}
    \quad
    \match{p}{M}{}
    \quad
    \negMatch{dp}{DM} }
    {D M \idev_{dp} D' M'}
    {\PCDataNoT} 
\onlyFull{\\}
\end{array}
$ \\

\end{center}
\end{definition}
\onlyEa{\vspace{-2mm}}
Remark  that rule \PCDataNoT\ is useful to deal with non-linear
patterns. \\
Thus for example, 
$ab ((\l y. y)c) \idev_{axx} abc$ since
$ab \dev ab$, $(\l y. y)c \dev c$, $\match{ax}{ab}{}$, 
$\match{x}{(\l y. y)c}{}$ but $\negMatch{axx}{ab((\l y. y)c)}$.

We observe also that if $N \idev N'$ or $N \idev_p N'$ then $N \dev N'$.

\onlyFull{
\bigskip
The following lemma analyses data / non-data preservation
\begin{lemma}[Development and data] \label{lem:dev-data} \abajito
\begin{enumerate}[a.]
	\item  \label{dev-date:uno} (internal development cannot create data terms) 
			Let $M \notin \setName{DataTerms}$, $N$ such that $M \idev N$. Then $N \notin \setName{DataTerms}$
	\item \label{dev-date:dos} (development from data produces always data) 
			Let $M \in \setName{DataTerms}$, $N$ such that $M \dev N$. Then $N \in \setName{DataTerms}$
\end{enumerate}
\end{lemma}
}

The formal description of the h-development condition takes a form of an additional binary relation.
This relation corresponds to the one called \textit{strong parallel reduction} in \cite{crary-std}.
\vspace{-2mm}
\begin{definition}[H-development]
We define the relations $\sepdev$ and $\subsepdev$. Let $M,N$ be terms; $\nu, \theta$ substitutions.
\vspace{-2mm}
\begin{enumerate}[a.]
	\item 
	$M \sepdev N$   \hspace{3mm}
	iff  \hspace{3mm}
	\parteuno\ $M \dev N$,  \hspace{3mm}
	\partedos\ $\metaexists Q \sthat M \hredn Q \idev N$,   \hspace{3mm}
	\partetres\ $\metaforall p \metadot \metaexists Q_p \sthat M \hRedPn{p} Q_p \idev_p N$.
	
	\item
	$\nu \subsepdev \theta$ \hspace{3mm}
	iff \hspace{3mm}
	\parteuno\ $Dom(\nu) = Dom(\theta)$, \hspace{3mm}
	\partedos\ $\metaforall x \in Dom(\nu) \metadot \nu x \sepdev \theta x $.
\end{enumerate}
\end{definition}

The clause \partetres\ in the definition of $\sepdev$ shows the
dependency on the patterns that was already noted in the definitions 
of  head step and
internal development. 

This clause is needed when proving that all developments are
h-developments; let's grasp the reason through a brief argument.
Suppose we want to prove that a development inside $N$ in a term like
$(\l p.M)N$ is an h-development. The rules to be used in this case are \HPat\ (Def.~\ref{dfn:hred})
and \IAppR\ (Def.~\ref{dfn:idev}). Therefore we need to perform an
analysis \textit{relative to the pattern $p$}; and this is exactly expressed
by clause \partetres. 
Consequently the proof of clause \partedos\ for a term needs to
consider clause \partetres\ (instantiated to a certain pattern) 
for a subterm; this is achieved by including clause \partetres\ in the
definition and by performing an inductive reasoning  on terms.

\section{Auxiliary results}
\label{s:prelimResults}
We collect in this section some results needed to complete the main proofs in this article.

\begin{lemma}[pattern-head reduction only if there is no match] \label{lem:hredp-only-if-negmatch} \abajo
Let $M, N$ be terms, $p$ a pattern, such that $M \hRedP{p} N$. Then $\negMatch{p}{M}$.
\end{lemma}
\begin{proof}
Using \reflem{head-step-basics}:(b).
\end{proof}

\begin{lemma}[development cannot lose matches] \label{lem:dev-cannot-lose-match} \abajo
Let $M, N$ be terms, $p$ a pattern, such that $M \dev N$ and $\match{p}{M}{\nu}$. Then $\match{p}{N}{\theta}$ for some $\theta$ such that $\nu \sdev \theta$.
\end{lemma}
\begin{proof}
Induction on $\match{p}{M}{\nu}$. The axioms can be checked trivially. For the rule, let $M = M_1 M_2$, $N = N_1 N_2$, $p = p_1 p_2$ and $\nu = \nu_1 \uplus \nu_2$
; $p$ is linear since it matches a term
. 
The only rules applicable for $M \dev N$ are \DRefl\ or \DApp; \DBeta\ is not applicable because $M_1 \in \setName{DataTerms}$.
If \DRefl\ was used, the lemma holds trivially taking $\theta = \nu$.
If \DApp\ was used, we apply the IH on both hypotheses obtaining $\match{p_i}{N_i}{\theta_i}$ with $\nu_i \sdev \theta_i$
; by \reflem{calculus-basics}:(\ref{bfc:e}) and the linearity of $p$ we know 
$\theta = \theta_1 \uplus \theta_2$ is well-defined; it is easy to check that $\theta$ satisfies the lemma conditions.
\end{proof}

\begin{lemma}[$\idev_{p}$ cannot create match] \label{lem:idevp-cannot-obtain-match} \abajo
Let $M, N$ be terms, $p$ a pattern, such that $M \idev_p N$.
Then $\negMatch{p}{M}$ implies $\negMatch{p}{N}$.
\end{lemma}

\begin{proof}
Induction on $M \idev_p N$ by rule analysis

\begin{description}
	\item[\PMatch] not applicable as $\negMatch{p}{M}$.
	\item[\PConst] in this case the condition $\negMatch{p}{M}$ implies $\negMatch{p}{N}$ equates to $M \neq p$ implies $N \neq p$, as $p$ is a constant. \\
	The rule premise reads $M \idev N$: if rule \IRefl\ was used then $N \neq p$ by hypothesis, else the \idev\ rule conclusions exclude the possibility of $N$ being a constant.
	\item[\PNoCData] $M \notin \setName{DataTerms}$ and $M \idev N$ by rule hyp., then $N \notin \setName{DataTerms}$ by L.~\ref{lem:dev-data}:(\ref{dev-date:uno}), finally $\negMatch{p}{N}$ 
        by L.~\ref{lem:calculus-basics}:(\ref{bfc:d}).
	\item[\PCDataNoL] By the IH, as rule hyp. includes both $D \idev_d D'$ and $\negMatch{d}{D}$ being $M = DT$ and $p = dp'$.
	\item[\PCDataNoR] Similar to the former considering $p = dp'$ and using $T \idev_{p'} T'$ and $\negMatch{p'}{T}$.	
        \item[\PCDataNoT] In this case $M = DM'$, $p = dp'$, $\match{d}{D}{\theta}$,
        $\match{p'}{M'}{\theta'}$ and $\negMatch{dp'}{DM'}$. We necessarily have 
        that $\theta  \uplus \theta'$ is not defined hence $p$ is not linear so that 
        $\negMatch{p}{N}$ also holds. 
\end{description} 
\end{proof}

\begin{lemma}[left-pattern-head implies whole-pattern-head] \label{lem:left-hredp-implies-hredp} \abajo
Let $p_1, p_2$ be patterns and $M_1, N_1, M_2$ be terms such that $M_1 \hRedP{p_1} N_1$.
Then $M_1 M_2 \hRedP{p_1 p_2} N_1 M_2$.
\end{lemma}

\begin{proof}
It is clear that $p_1 \notin Var$, because there is no $N_1$ such that $M_1 \hRedP{x} N_1$ if $x \in Var$.

If \PatHead\ applied in $M_1 \hRedP{p_1} N_1$, then $M_1 \hred N_1$, by \HAppL\ $M_1 M_2 \hred N_1 M_2$, and finally by \PatHead\ $M_1 M_2 \hRedP{p_1 p_2} N_1 M_2$.

If either \PatL\ or \PatR\ applied in $M_1 \hRedP{p_1} N_1$, then $M_1$ is clearly a data term, 
Then $M_1 M_2 \hRedP{p_1 p_2} N_1 M_2$  by \PatL.
\end{proof}

\begin{lemma}[matching is compatible with substitution] \abajo
\label{lem:match-compatible-subst}
Let $M$ be a term, $p$ a pattern and  $\theta$ a substitution such that $\match{p}{M}{\theta}$.
Then for any substitution $\nu$, the following holds: 
$\match{p}{\nu M}{\gamma}$ 
where $\gamma = \nu \theta \mid_{\fv{p}}$.
\end{lemma}

\begin{proof}
By induction on the match. The axioms can be checked trivially given L.~\ref{lem:calculus-basics}:(\ref{bfc:e}). 

We analyze the rule applied in this context
$$
\frac
{\match{d}{M_1}{\theta_1} \qquad \match{p'}{M_2}{\theta_2}}
{\match{dp' = p}{M = M_1 M_2}{\theta = \theta_1 \uplus \theta_2}}
$$
Applying the IH on both hypotheses and then using the rule gives 
$\match{dp'}{M_1 M_2}{(\nu \theta_1) \mid_{\fv{d}} \uplus (\nu \theta_2) \mid_{\fv{p}}}$; an easy check of 
$(\nu \theta_1) \mid_{\fv{d}} \uplus (\nu \theta_2) \mid_{\fv{p'}} 
  = 
 (\nu (\theta_1 \uplus \theta_2)) \mid_{\fv{dp'}}$ 
concludes the proof.
\end{proof}

\begin{lemma}[development is compatible with substitution] \label{lem:dev-compatible-subst} \abajo
Let $M, N$ be terms and  $\nu, \theta$ substitutions, such that $M \dev N$ and $\nu \sdev \theta$. 
Then $\nu M \dev \theta N$ 
\end{lemma}

\begin{proof}
By induction on $M \dev N$ by rule analysis.

For \DRefl\ the thesis amounts to $\nu M \dev \theta M$, which can be checked by a simple induction on $M$.
\DAbs\ and \DApp\ can be simply verified by the IH.

For \DBeta\ first we mention a technical result which will be used.
Let $\theta$, $\tau$ be substitutions such that $\dom(\tau) \cap \var{\theta} = \emptyset$, then 
\begin{equation}
	\big( (\theta \tau) \mid_{\dom(\tau)} \big) \: \theta = \theta \tau
	\label{eq:technical-substitution-result}
\end{equation}
this can be easily checked comparing the effect of applying both substitutions to an arbitrary variable.

Let's analyze the rule premises and conclusion applied in this context
$$\frac
  {M_1 \dev M_1' \qquad \tau \sdev \tau' \qquad \match{p}{M_2}{\tau}}
  {M = (\l p.M_1) M_2 \dev \tau' M_1' = N} 
$$
As we can freely choose the variables appearing in $p$, we assume
$\fv{p} \cap (\var{\nu} \cup \var{\theta}) = \emptyset$. 
By L.~\ref{lem:calculus-basics}:(\ref{bfc:e}) we know $\dom(\tau) = \dom(\tau') = \fv{p}$.

We apply the IH on $M_1 \dev M_1'$ and also on $\tau x \dev \tau' x$ for each $x \in \dom(\tau)$ to conclude 
$\nu M_1 \dev \theta M_1'$ and 
$(\nu \tau) \mid_{\dom(\tau)} \sdev (\theta \tau') \mid_{\dom(\tau)}$
respectively.
Furthermore, from $\match{p}{M_2}{\tau}$ and L.~\ref{lem:match-compatible-subst} we conclude 
$\match{p}{\nu M_2}{(\nu \tau \mid_{\dom(\tau)})}$.

We use \DBeta\ from the three conclusions above to obtain
$$\nu M = (\l p.\nu M_1) (\nu M_2) 
  \dev 
  \big( (\theta \tau') \mid_{\dom(\tau)} \big) (\theta M_1')$$
To check
$\theta N = \theta (\tau' M_1') = 
 \big( (\theta \tau') \mid_{\dom(\tau)} \big) (\theta M_1')$ 
it is enough to verify \\
$\theta \tau' = \big( (\theta \tau') \mid_{\dom(\tau)} \big) \theta$,
the latter can be easily checked by (\ref{eq:technical-substitution-result}).

\end{proof}

\begin{lemma}[head reduction is compatible with substitution] \abajito
\label{lem:hred-compatible-subst}
\begin{itemize}
	\item[\parteuno] 
	Let $M, N$ be terms and  $\nu$ a substitution such that $M \hred N$. 
	Then $\nu M \hred \nu N$.
	
	\item[\partedos] 
	Let $M, N$ be terms, $p$ a pattern and  $\nu$ a substitution such that $M \hRedP{p} N$. 
	Then $\nu M \hRedP{p} \nu N$.
\end{itemize}
\end{lemma}

\begin{proof} (sketch) \\
Both items are proved by simultaneous induction on $M \hred N$ and $M \hRedP{p} N$.

We use \reflem{match-compatible-subst} for case \HBeta, the IH
and \reflem{match-compatible-subst} for case \PatR, and
just the IH for the remaining cases.
\end{proof}

\section{H-developments}
\label{s:separability}
The aim of this section is to prove that all developments are h-developments. 

We found easier to prove separately that the h-development condition is compatible with the language constructs, diverging from the structure of the proofs in \cite{crary-std}.

\begin{lemma}[$\sepdev$ is compatible with abstraction] \abajo
\label{lem:sepdev-compatible-abs}
Let $M, N$ be terms such that $M \sepdev N$. Then $\l q.M \sepdev \l q.N$ for any pattern $q$.
\end{lemma}

\begin{proof} 
	Part \parteuno\ trivially holds  by hyp. \parteuno\ and \DAbs.

	Part \partedos: by hyp. \parteuno\ and \IAbs\ we get $\l q.M \idev \l q.N$. Then $Q = \l q.M$.
	
	Part \partetres:  if $p \in Var$ then \PMatch\ applies, if $p$ is a constant or a compound data pattern then \PConst\ or \PNoCData\ apply respectively  as $(\l q.M) \idev (\l q.N)$. In all cases we obtain $(\l q.M) \idev_p (\l q.N)$.
        Then $Q = \l q.M$.	
\end{proof}

\begin{lemma}[$\sepdev$ is compatible with application] \abajo
\label{lem:sepdev-compatible-app}
Let $M_1, M_2, N_1, N_2$ be terms such that $M_1 \sepdev N_1$ and $M_2 \sepdev N_2$. 
Then $M_1 M_2 \sepdev N_1 N_2$.
\end{lemma}

\begin{proof}
	Part \parteuno\ is immediate by the hypotheses \parteuno\ and \DApp.

	Let's prove part \partedos.
	
	We first use hypothesis \partedos\ on $M_1 \sepdev N_1$ to obtain $M_1 \hredn Q_1 \idev N_1$ and subsequently apply \HAppL\ to $M_1 \hredn Q_1$ to get
	\begin{eqnarray}
		M_1 M_2 & \hredn & Q_1 M_2     \label{eq:comp-app-51}
	\end{eqnarray}
	Either $Q_1$ is an abstraction or not.
	
	\medskip
	Assume $Q_1$ is not an abstraction. Since $Q_1 \idev N_1$ and $M_2 \dev N_2$,  we apply \IAppL\ so that $Q_1 M_2 \idev N_1 N_2$; this together with \refeq{comp-app-51} gives the desired result.

	\medskip
	Now assume $Q_1 = \l p.Q_{12}$.	
	We use the hyp. \partetres\ on $M_2 \sepdev N_2$, obtaining $M_2 \hRedP{p}^* Q_2 \idev_p N_2$ 
	and then we apply \HPat\ to get
	\begin{eqnarray}
		Q_1 M_2 & \hredn & Q_1 Q_2					\label{eq:comp-app-61}
	\end{eqnarray} 
	Moreover, as $Q_1 = \l p.Q_{12} \idev N_1$, the only applicable rules are \IRefl\ or \IAbs, and in both cases $N_1 = \l p.N_{12}$ and $Q_{12} \dev N_{12}$. \\
	We now use \IAppR\ with premises $Q_{12} \dev N_{12}$ and $Q_2 \idev_p N_2$ to get 
	\begin{eqnarray}
		Q_1 Q_2 = (\l p.Q_{12}) Q_2 & \idev & (\l p.N_{12}) N_2 = N_1 N_2
		\label{eq:comp-app-62}
	\end{eqnarray} 
	The desired result is obtained by \refeq{comp-app-51}, \refeq{comp-app-61} and \refeq{comp-app-62}.
	
	\bigskip
	Let's prove part \partetres. 

	\medskip	
	If $p \in Var$ we are done by \parteuno\ and \PMatch; we thus get 
        $M_1 M_2 \idev_p N_1 N_2$ so that $Q = M_1 M_2$.

	\medskip	
	If $p = c$ then using \partedos\ we obtain
	$M_1 M_2 \hredn Q \idev N_1 N_2$ for some $Q$ 
	; we apply \PatHead\ and \PConst\ to get $M_1 M_2 \hRedPn{c} Q$ and $Q \idev_c N_1 N_2$ respectively, concluding the proof for this case.
	
	\bigskip
	Consider $p = p_1 p_2$ with $p_1$ a data pattern and $p_2$ a pattern.

  We use the hyp.  \partetres\ on $M_1 \sepdev N_1$, getting 
			$M_1 \hRedPn{p_1} Q_1 \idev_{p_1} N_1$.
	Let us define $R_1$ as follows: if there is a data term in the sequence $M_1 \hRedPn{p_1} Q_1$
        then $R_1$ is the first of such terms; 
        otherwise $R_1$ is $Q_1$. In both cases $M_1 \hRedPn{p_1} R_1 \hRedPn{p_1} Q_1$. We necessarily have
         $M_1 \hredn R_1$ by \PatHead, then  $M_1 M_2 \hredn R_1 M_2$ by \HAppL\ 
        and subsequently  $M_1 M_2 \hRedP{p} R_1 M_2$ by \PatHead.
	
	We conclude $M_1 M_2 \hRedPn{p} Q_1 M_2$, trivially if $Q_1 = R_1$, and applying \PatL\ to $R_1 \hRedPn{p_1} Q_1$ to obtain $R_1 M_2 \hRedPn{p} Q_1 M_2$ otherwise.

	\medskip
	If $Q_1 = (\l q.Q_1')$ then we use the hyp.  \partetres\ on $M_2 \sepdev N_2$ getting $M_2 \hRedPn{q} Q_2 \idev_{q} N_2$.
	
	We apply \HPat\ to $M_2 \hRedP{q}^* Q_2$ getting $Q_1 M_2 \hredn Q_1 Q_2$; therefore we obtain $Q_1 M_2 \hRedPn{p} Q_1 Q_2$ by  \PatHead.
		
	In the other side $Q_1 = (\l q.Q_1') \dev N_1$, therefore $N_1 = (\l q.N_1')$ and $Q_1' \dev N_1'$.
	
	We apply \IAppR\ to $Q_1' \dev N_1'$ and $Q_2 \idev_{q} N_2$ to obtain $Q_1 Q_2 \idev N_1 N_2$, therefore $Q_1 Q_2 \idev_p N_1 N_2$ by \PNoCData. We thus get the desired result taking $Q_p = Q_1 Q_2$.
	
	\medskip
	If $Q_1$ is not an abstraction and $Q_1 \notin \setName{DataTerms}$, then only \PConst\ or 
        \PNoCData\ can justify $Q_1 \idev_{p_1} N_1$, thus implying $Q_1 \idev N_1$; this together with
        the hypothesis \parteuno\ $M_2 \dev N_2$ gives $Q_1 M_2 \idev N_1 N_2$
        by \IAppL, hence  $Q_1 M_2 \idev_p N_1 N_2$ by \PNoCData. We get the desired result by taking $Q_p = Q_1 M_2$.

	\medskip
	If $Q_1 \in \setName{DataTerms}$ we anaylise the different 
        alternatives for the matching between $p_1 p_2$ and $Q_1 M_2$.
	
	Assume $\negMatch{p_1}{Q_1}$. In this case we apply \PCDataNoL\ to 
	$Q_1 \idev_{p_1} N_1$ and $M_2 \dev N_2$ to obtain $Q_1 M_2 \idev_p N_1 N_2$ and thus the 
        desired result holds by taking $Q_p = Q_1 M_2$.
	
	Assume $\match{p_1}{Q_1}{}$ and $\negMatch{p_2}{M_2}$.
	In this case we use the hyp. \partetres\ on $M_2 \sepdev N_2$ to get
	$M_2 \hRedPn{p_2} Q_2 \idev_{p_2} N_2$, then apply \PatR\ to get 
        $Q_1 M_2 \hRedPn{p} Q_1 Q_2$. Finally from $Q_1 \idev_{p_1} N_1$ and $Q_2 \idev_{p_2} N_2$ we 
        obtain $Q_1 Q_2 \idev_p N_1 N_2$ by either  \PCDataNoR, \PCDataNoT\ or \PMatch. 
       We get the desired result by taking $Q_p = Q_1 Q_2$.
	
	Finally assume $\match{p_1}{Q_1}{}$ and $\match{p_2}{Q_2}{}$.
	In this case the hypotheses imply in particular $Q_1 \dev N_1$ and $M_2 \dev N_2$. We thus conclude $Q_1 M_2 \idev_{p} N_1 N_2$ using either \PMatch\ or \PCDataNoT\ (depending on whether $\match{p}{Q_1 M_2}{}$ or not), getting the desired result by taking $Q_p = Q_1 M_2$.

\end{proof}

\bigskip
Now we proceed with the proof of the h-development property. The
generalization of the statement involving $\subsepdev$ is needed to
conclude the proof\footnote{In \cite{crary-std} the compatibility of
  h-development with substitutions is stated as a separate lemma; for
  pattern calculi we could not find a proof of compatibility with
  substitution independent of the main h-development result.}, as
can be seen in the \DBeta\ case below.

\begin{lemma}[Generalized h-developments property] \abajo
\label{lem:crary-6-generalized}
Let $M, N$ be terms and $\nu, \theta$ substitutions, such that $M \dev N$ and $\nu \subsepdev \theta$. \\
Then $\nu M \sepdev \theta N$
\end{lemma}

\begin{proof}
By induction on $M \dev N$ analyzing the rule used in the last step of the derivation. 

\begin{description}
	\item[DRefl] 
	in this case $N = M$, we proceed by induction on $M$
	\newcommand{\sepdevamplio}{\; \sepdev \;}
	\begin{itemize}
		\item 
		$M = x \in Dom(\nu)$, in this case $\nu M = \nu x \sepdevamplio \theta x = \theta N$ by hypothesis.
		
		\item 
		$M = x \notin Dom(\nu)$, in this case $\nu M = x \sepdevamplio x = \theta N$.
		
		\item 
		$M = M_1 M_2$, in this case$\nu M_1 \sepdev \theta M_1$ and $\nu M_2 \sepdev \theta M_2$
                hold by the IH. The desired result is obtained by \reflem{sepdev-compatible-app}.
		
		\item
		$M = \l p.M_1$, in this case $\nu M_1 \sepdev \theta M_1$ holds by the IH. The desired result 
                is obtained by \reflem{sepdev-compatible-abs}.
	\end{itemize}

	\item[DAbs]
	in this case $M = \l p.M_1, N = \l p.N_1,  M_1 \dev N_1$. 

	Using the IH on $M_1 \dev N_1$ we obtain $\nu M_1 \sepdev \theta N_1$, the desired result is obtained  by \reflem{sepdev-compatible-abs}.
		
	\item[DApp] 
	in this case $M = M_1 M_2, N = N_1 N_2, M_i \dev N_i$.
	
	Using the IH on both rule premises we obtain $\nu M_i \sepdev \theta N_i$, the desired result is obtained  by \reflem{sepdev-compatible-app}.
	
	\item[DBeta] 
	Let's write down the rule instantiation
	
	$\regla
		{ M_{12} \dev N_{12}   \quad
		  \tau \sdev \tau'   \quad
		  \match{q}{M_2}{\tau} }
		{ M = (\l q.M_{12}) M_2 \,\dev\, \tau' N_{12} = N }
		{}
	$	 
	
	\parteuno\ can be obtained by hypotheses $M \dev N$ and $\nu \subsepdev \theta$, and then \reflem{dev-compatible-subst}.
	
	For [ \partetres\ if $p \in Var$ ] we are done by \parteuno\  and \PMatch.
	
	For [ \partetres\ if $p = d$ ] and also for \partedos\ : we know both 
	$ M \hred \tau M_{12}$ and $M \hRedP{p} \tau M_{12} $,
	then by \reflem{hred-compatible-subst}
	\begin{eqnarray}
	\nu M \hred \nu (\tau M_{12}) & \textand & \nu M \hRedP{p} \nu (\tau M_{12})
	\label{eq:lemma6-dbeta-01}
	\end{eqnarray}
	
	We apply the IH on each $\tau x \dev \tau' x$, obtaining $(\nu \tau) x = \nu (\tau x) \sepdevamplio \theta (\tau' x) = (\theta \tau') x$ for all $x \in Dom(\tau)$. 
	Moreover, if $x \in Dom(\nu) - Dom(\tau)$ then 
        $(\nu \tau) x = \nu x \sepdevamplio \theta x = (\theta \tau') x$  by hypothesis.
	
	Consequently, $\nu \tau \subsepdev \theta \tau'$. Now we  use the  IH on $M_{12} \dev N_{12}$ taking $\nu \tau \subsepdev \theta \tau'$ as second hypothesis to obtain
	$$ \nu (\tau M_{12}) = (\nu \tau) M_{12} \sepdevamplio (\theta \tau') N_{12} = \theta (\tau' N_{12}) = \theta N $$
	
	This result along with \refeq{lemma6-dbeta-01} concludes the proof for both parts.
\end{description}
\end{proof}

\begin{corollary}[H-development property] \abajo
\label{lem:crary-6}
Let $M, N$ be terms such that $M \dev N$. 
Then $M \sepdev N$.
\end{corollary}

\section{Standardisation}
\label{s:standardisation}
The part of the standardisation proof following the proof of the h-development property coincides in structure with the proof given in \cite{crary-std}.

\medskip
First we will prove that we can get, for any reduction involving head steps that follows an internal development, another reduction in which the head steps are at the beginning. The name given to the Lemma~\ref{lem:crary-7} was taken from \cite{crary-std}.

This proof needs again to consider explicitly the relations relative to patterns, for similar reasons to those described when introducing h-development in section \ref{s:concepts}.

\begin{lemma}[Postponement] \hspace{1cm}
\label{lem:crary-7}
\begin{itemize}
	\item[\parteuno\ ] if $M \idev N \hred R$ 
	           then there exists some term $N'$
	           such that $M \hred N' \dev R$
	\item[\partedos\ ] for any pattern $p$, 
	           if $M \idev_p N \hRedP{p} R$ 
	           then there exists some term $N'_p$ 
	           such that $M \hRedP{p} N'_p \dev R$ 
\end{itemize}
\end{lemma}

\begin{proof}
For \parteuno, if the rule used in $M \idev N$ is \IRefl, then the result is immediate taking $N' = R$. Therefore, in the following we will ignore this case.

We prove \parteuno\ and \partedos\ by simultaneous induction on $M$ taking into account the previous observation. 

\begin{description}
	\item[variable] in this case it must be $N = M$ for both \parteuno\ and \partedos\, and neither $M \hred R$ nor $M \hRedP{p} R$ for any $p, R$.

	\item[abstraction] in this case $N$ must also be an abstraction for both \parteuno\ and \partedos\, and neither $N \hred R$ nor $N \hRedP{p} R$ for any $p, R$.

	\item[application] 
	in this case $M = M_1 M_2$
	
	We prove \parteuno\ first, analysing the possible forms of $M_1$
	
	\begin{itemize}
		\item 
		Assume $M_1$ is not an abstraction
		
		In this case \IAppL\ applies, so we know $N = N_1 N_2$, $M_1 \idev N_1$, and $M_2 \dev N_2$.
		
		Since $M_1 \idev N_1$, $N_1$ is not an abstraction, then the only applicable rule for $N \hred R$ is \HAppL, hence $R = R_1 N_2$ and $N_1 \hred R_1$.
		
		Now we use the IH on $M_1 \idev N_1 \hred R_1$ to get $M_1 \hred N_1' \dev R_1$, then we obtain $M = M_1 M_2 \hred N_1' M_2$ by \HAppL.
		
		Finally we apply \DApp\ to $N_1' \dev R_1$ and $M_2 \dev N_2$ to get $N_1' M_2 \dev R_1 N_2 = R$, which concludes the proof for this case.

		\item 
		Now assume $M_1 = \l p.M_{12}$ and $\negMatch{p}{M_2}$
		
		Since $M = (\l p.M_{12}) M_2 \idev N$, the only rule
                that applies is \IAppR, then $N = (\l p.N_{12}) N_2$,
                $M_{12} \dev N_{12}$, and $M_2 \idev_p N_2$. By
                \reflem{idevp-cannot-obtain-match} we obtain
                $\negMatch{p}{N_2}$, so the only applicable rule in
                $N = (\l p.N_{12}) N_2 \hred R$ is \HPat, then $R =
                (\l p.N_{12}) R_2$ and $N_2 \hRedP{p} R_2$.
		
		Now we use the IH \partedos\ on $M_2 \idev_p N_2 \hRedP{p} R_2$, to get $M_2 \hRedP{p} N_2' \dev R_2$.
		
		We obtain $M = (\l p.M_{12}) M_2 \hred (\l p.M_{12}) N_2'$ by \HPat, then we get $(\l p.M_{12}) \dev (\l p.N_{12})$ by \DAbs\ on $M_{12} \dev N_{12}$, 
		finally we apply \DApp\ to the previous result and $N_2' \dev R_2$ 
		to obtain $(\l p.M_{12}) N_2' \dev (\l p.N_{12}) R_2 = R$ which concludes the proof for this case.

		\item 
		Finally, assume $M_1 = \l p.M_{12}$ and $\match{p}{M_2}{\nu}$
		
		Again, the only rule that applies  in $M = (\l
                p.M_{12}) M_2 \idev N$ is \IAppR, then $N = (\l
                p.N_{12}) N_2$, $M_{12} \dev N_{12}$, and $M_2 \idev_p
                N_2$.  Now, by \reflem{dev-cannot-lose-match} we
                obtain $\match{p}{N_2}{\theta}$
                for some substitution $\theta$ such that $\nu \sdev
                \theta$, then the applied rule in $N \hred R$ is
                \HBeta\ (the case \HPat\ being excluded by
                \reflem{hredp-only-if-negmatch}), hence $R = \theta
                N_{12}$
		
		It is clear that $M \hred \nu M_{12}$. By
                \reflem{dev-compatible-subst} we obtain $\nu M_{12} \dev \theta N_{12} = R$,
                which concludes the proof for this case.
	\end{itemize}

	\bigskip
	For \partedos\ we proceed by a case analysis of $p$ 
	
	\medskip	
	If $p \in Var$ then there is no $R$ such that $N \hRedP{p} R$ for any term $N$. 

	\medskip
	If $\match{p}{M}{}$ then by \reflem{dev-cannot-lose-match} $\match{p}{N}{}$, and therefore by \reflem{hredp-only-if-negmatch} there can be no $R$ such that $N \hRedP{p} R$.
	
	\medskip	
	If $p = c$ then $\negMatch{p}{M}$, hence $M \idev_p N
        \hRedP{p} R$ implies $M \idev N \hred R$ as \PConst\ and
        \PatHead\ are the only possibilities for this case
        respectively. We use part \parteuno\ to obtain $M
        \hred N' \dev R$, and $M \hRedP{p} N'$ by \PatHead\ which
        concludes the proof for this case.
	
	\medskip
		If $p = d \, p_2$ and $M \notin \setName{DataTerms}$, then the only possibilities for $M \idev_p N \hRedP{p} R$ are \PNoCData\ and \PatHead\ respectively, then $M \idev N \hred R$.
		We use part \parteuno\  to obtain $M \hred N' \dev R$, and $M \hRedP{p} N'$ by \PatHead\ which concludes the proof for this case.
	
	\medskip
	Now assume $p = d \, p_2$, $M \in \setName{DataTerms}$, and $\negMatch{p}{M}$. We must analyse three possibilities 
	
	\begin{itemize}
		\item 
		$\negMatch{d}{M_1}$. \\
		In this case only \PCDataNoL\ applies  for $M \idev_p N$, therefore $N = N_1 N_2$ with $M_1 \idev_{d} N_1$ and $M_2 \dev N_2$. By \reflem{idevp-cannot-obtain-match} we know $\negMatch{d}{N_1}$ and moreover $N_1$ is a data term (as can be seen by \reflem{dev-data}) thus not having head redexes, so the only possible rule for $N \hRedP{p} R$ is \PatL, then $R = R_1 N_2$ with $N_1 \hRedP{d} R_1$.
				
		Now we use the IH on the derivation $M_1 \idev_{d} N_1 \hRedP{d} R_1$ to get $M_1 \hRedP{d} N_1' \dev R_1$, therefore $M = \appfrom{M} \hRedP{p} N_1' M_2$ by \PatL.
		
		Moreover as $N_1' \dev R_1$ and $M_2 \dev N_2$ hence $N_1' M_2 \dev R_1 N_2 = R$, which concludes the proof for this case.
		
		\item 
		$\match{d}{M_1}{}$ and $\negMatch{p_2}{M_2}$. \\ In
                  this case only \PCDataNoR\ applies  for $M \idev_p
                  N$, therefore $N = N_1 N_2$ with $M_1 \dev N_1$ and
                  $M_2 \idev_{p_2} N_2$.  By
                  \reflem{dev-cannot-lose-match} and
                  \reflem{idevp-cannot-obtain-match} respectively, we
                  obtain both $\match{d}{N_1}{}$ and
                  $\negMatch{p_2}{N_2}$.  Moreover $N$ is a data term
                  (as can be seen by \reflem{dev-data}) thus not
                  having head redexes. Hence the only possibility for
                  $N \hRedP{p} R$ is \PatR, then $R = N_1 R_2$ with
                  $N_2 \hRedP{p_2} R_2$
		
	We now use the IH on $M_2 \idev_{p_2} N_2 \hRedP{p_2} R_2$ to get $M_2 \hRedP{p_2} N_2' \dev R_2$, and by \PatR\ $M = \appfrom{M} \hRedP{p} M_1 N_2'$
		
		We also use \DApp\ on $M_1 \dev N_1$ and $N_2' \dev R_2$ to get $M_1 N_2' \dev N_1 R_2 = R$, which concludes the proof for this case.

    \item 
		$\match{d}{M_1}{}$, $\match{p_2}{M_2}{}$ and $\negMatch{dp_2}{M_1M_2}$. 
             
                $\match{d}{M_1}{}$ implies
                  (L~\ref{lem:calculus-basics}:(\ref{bfc:d})) $M_1 \in
                  \setName{DataTerms}$ so that from $M = M_1 M_2
                  \idev_p N$ we can only have $N = N_1 N_2$ with $M_1
                  \dev N_1$ and $M_2 \dev N_2$.
                  L.~\ref{lem:dev-cannot-lose-match} gives
                  $\match{d}{N_1}{}$ and $\match{p_2}{N_2}{}$.
                  L.~\ref{lem:dev-data}:(\ref{dev-date:dos}) gives
                  $N \in 
                  \setName{DataTerms}$. To show 
                  $N \hRedP{p} R$ we have three possibilities:
                  $\PatHead$ is not possible since $N \in 
                  \setName{DataTerms}$ (c.f. L~\ref{lem:head-step-basics}:(\ref{bf:a})), 
                  $\PatL$ is not possible since $\match{d}{M_1}{}$ (c.f. L~\ref{lem:hredp-only-if-negmatch}),
                  $\PatR$ is not possible since $\match{p_2}{N_2}{}$
                   (c.f. L~\ref{lem:hredp-only-if-negmatch}).
	\end{itemize}

\end{description}
\end{proof}

\begin{corollary}
\label{lem:crary-8} \abajo
Let $M, N, R$ be terms such that $M \idev N \hred R$. Then $\metaexists N' \sthat 
M \hredn N' \idev R$.
\end{corollary}

\begin{proof}
Immediate by \reflem{crary-7} and Corollary~\ref{lem:crary-6}.
\end{proof}

\bigskip\noindent
Now we generalize the h-development concept to a sequence of developments. The name given to Lemma~\ref{lem:crary-9} was taken from \cite{crary-std}.

\begin{lemma}[Bifurcation] 
\label{lem:crary-9} \abajo
Let $M, N$ be terms such that $M \devn N$. Then $M \hredn R \idevn N$ for some term $R$.
\end{lemma}

\begin{proof}
Induction on the length of $M \devn N$. If $M = N$ the result holds trivially.

Assume $M \dev Q \devn N$. By C.~\ref{lem:crary-6} and IH respectively, we obtain
$M \hredn S \idev Q$ and $Q \hredn T \idevn N$ for some terms $S$ and $T$. 
Now we  use Corollary~\ref{lem:crary-8} (many times) on $S \idev Q \hredn T$ to get $S \hredn R \idev T$.

Therefore $M \hredn S \hredn R \idev T \idevn N$ as we desired.
\end{proof}

\bigskip\noindent
Using the previous results, the standardisation theorem admits a very simple proof.

\begin{definition}[Standard reduction sequence]
The standard reduction sequences 
are the sequences of terms $M_1;\ldots;M_n$ which can be generated using the following rules.
\begin{center}
$
\begin{array}{c}
\reglaalta
    {M_2;\ldots;M_k
    \quad
    M_1\hred M_2}
    {M_1;\ldots;M_k}
    {\schemeName{StdHead}}
\hspace{1cm}
\regla
    {M_1;\ldots;M_k}
    {(\l p.M_1);\ldots;(\l p.M_k)}
    {\schemeName{StdAbs}}
\\
\\
\regla
    {M_1;\ldots;M_j
    \quad
     N_1;\ldots;N_k}
    {(M_1\,N_1);\ldots (M_j\,N_1);(M_j\,N_2);\ldots;(M_j\,N_k)}
    {\schemeName{StdApp}}
\hspace{1cm}
\regla
    {}
    {x}
    {\schemeName{StdVar}}
\end{array}
$
\end{center}
\end{definition}

\begin{theorem}[Standardisation] 
\label{lem:crary-12} \abajo
Let $M, N$ be terms such that $M \devn N$. Then there exists a standard reduction sequence $M;\ldots;N$.
\end{theorem}

\begin{proof}
By \reflem{crary-9} we have $M \hredn R \idevn
N$; we observe that it is enough to obtain a standard reduction
sequence $R;\ldots;N$, because we subsequently apply
\schemeName{StdHead} many times.

Now we proceed by induction on $N$

\begin{itemize}
	\item $N \in Var$; 
	in this case $R = N$ and we are done.
	
	\item $N = \l p.N_1$; in this case $R = \l p.R_1$ and $R_1
          \devn N_1$.  By IH we obtain a standard reduction sequence $R_1;\ldots;N_1$, then
          by \schemeName{StdAbs} so is $R = \l p.R_1;\ldots;\l p.N_1 =
          N$.
	
	\item $N = N_1 N_2$, so $R = R_1 R_2$ and $N_i \devn R_i$. We use the IH on both reductions to get two standard reduction sequences $N_i;\ldots;R_i$, then we join them using \schemeName{StdApp}.
\end{itemize}
\end{proof}

\newsection{Conclusion and further work}
\label{s:conclusion}
We have presented an elegant proof of
the Standardisation Theorem for constructor-based pattern calculi.

We aim to generalize both the concept of standard reduction
and the structure of the Standardisation Theorem proof presented here
to a large class of pattern calculi, including both open and closed
variants as the Pure Pattern Calculus~\cite{JK05}.  It would be
interesting to have sufficient conditions for a pattern calculus to
enjoy the standardisation property. This will be close in spirit with
\cite{jk-fcp} where an abstract confluence proof for pattern calculi
is developed.

The kind of calculi we want to deal with imposes challenges 
that are currently not
handled in the present contribution, such as open patterns, reducible
(dynamic) 
patterns, and the possibility of having \texttt{fail} as a decided
result of matching.
Furthermore, the possibility of decided \texttt{fail} combined with
compound patterns leads to the convenience of studying forms of
\textit{inherently parallel} standard reduction strategies.

The 
abstract axiomatic Standardisation  Theorem   developed in~\cite{GLM92} 
could be useful for our purpose.   However, while the
axioms of  the abstract formulation of standardisation  are assumed to
hold  in the  proof of  the standardisation  result, they  need  to be
defined and verified  for each  language to  be standardised.   This could  be 
nontrivial,  as   in  the  case  of   TRS~\cite{HL91,Terese03},  where  a
meta-level  matching operation is  involved in  the definition  of the
rewriting framework. We leave this topic as further work.

\vspace{-4mm}

\bibliography{tesisd-resumido}

\begin{thebibliography}{10}
\providecommand{\bibitemstart}[1]{\bibitem{#1}}
\providecommand{\bibitemend}{}
\providecommand{\bibliographystart}{}
\providecommand{\bibliographyend}{}
\providecommand{\url}[1]{\texttt{#1}}
\providecommand{\urlprefix}{Available at }
\providecommand{\bibinfo}[2]{#2}
\bibliographystart

\bibitemstart{barendregt}
\bibinfo{author}{H.P. Barendregt} (\bibinfo{year}{1984}):
  \emph{\bibinfo{title}{The Lambda Calculus: Its Syntax and Semantics}}.
\newblock \bibinfo{publisher}{Elsevier}, \bibinfo{address}{Amsterdam}.
\bibitemend

\bibitemstart{crary-std}
\bibinfo{author}{K.~Crary} (\bibinfo{year}{2009}): \emph{\bibinfo{title}{A
  Simple Proof of Call-by-Value Standardization}}.
\newblock \bibinfo{type}{Technical Report} \bibinfo{number}{CMU-CS-09-137},
  \bibinfo{institution}{Carnegie-Mellon University}.
\bibitemend

\bibitemstart{curry-feys}
\bibinfo{author}{H.B. Curry} \& \bibinfo{author}{R.~Feys}
  (\bibinfo{year}{1958}): \emph{\bibinfo{title}{Combinatory Logic}}.
\newblock \bibinfo{publisher}{North-Holland Publishing Company},
  \bibinfo{address}{Amsterdam}.
\bibitemend

\bibitemstart{girard}
\bibinfo{author}{J.-Y. Girard} (\bibinfo{year}{1987}):
  \emph{\bibinfo{title}{Linear Logic}}.
\newblock {\sl \bibinfo{journal}{Theoretical Computer Science}}
  \bibinfo{volume}{50}(\bibinfo{number}{1}), pp. \bibinfo{pages}{1--101}.
\bibitemend

\bibitemstart{GLM92}
\bibinfo{author}{G.~Gonthier}, \bibinfo{author}{J.-J. L{\'e}vy} \&
  \bibinfo{author}{P.-A. Melli{\`e}s} (\bibinfo{year}{1992}):
  \emph{\bibinfo{title}{An abstract standardisation theorem}}.
\newblock In: {\sl \bibinfo{booktitle}{Proceedings, Seventh Annual IEEE
  Symposium on Logic in Computer Science, 22-25 June 1992, Santa Cruz,
  California, USA}}, \bibinfo{publisher}{IEEE Computer Society}, pp.
  \bibinfo{pages}{72--81}.
\bibitemend

\bibitemstart{HL91}
\bibinfo{author}{G.~Huet} \& \bibinfo{author}{J.-J. L\'evy}
  (\bibinfo{year}{1991}): \emph{\bibinfo{title}{Computations in orthogonal
  rewriting systems}}.
\newblock In: \bibinfo{editor}{Jean-Louis Lassez} \& \bibinfo{editor}{Gordon
  Plotkin}, editors: {\sl \bibinfo{booktitle}{Computational Logic, Essays in
  Honor of Alan Robinson}}, \bibinfo{publisher}{MIT Press}, pp.
  \bibinfo{pages}{394--443}.
\bibitemend

\bibitemstart{JK05}
\bibinfo{author}{C.B. Jay} \& \bibinfo{author}{D.~Kesner}
  (\bibinfo{year}{2006}): \emph{\bibinfo{title}{Pure Pattern Calculus}}.
\newblock In: \bibinfo{editor}{Peter Sestoft}, editor: {\sl
  \bibinfo{booktitle}{European Symposium on Programming}}, number
  \bibinfo{number}{3924} in \bibinfo{series}{LNCS},
  \bibinfo{publisher}{Springer-Verlag}, pp. \bibinfo{pages}{100--114}.
\bibitemend

\bibitemstart{jk-fcp}
\bibinfo{author}{C.B. Jay} \& \bibinfo{author}{D.~Kesner}
  (\bibinfo{year}{2009}): \emph{\bibinfo{title}{First-class patterns}}.
\newblock {\sl \bibinfo{journal}{Journal of Functional Programming}}
  \bibinfo{volume}{19}(\bibinfo{number}{2}), pp. \bibinfo{pages}{191--225}.
\bibitemend

\bibitemstart{kashima-std}
\bibinfo{author}{Ryo Kashima} (\bibinfo{year}{2000}): \emph{\bibinfo{title}{A
  Proof of the Standardization Theorem in $\lambda$-Calculus}}.
\newblock \bibinfo{type}{Research Reports on Mathematical and Computing
  Sciences} \bibinfo{number}{C-145}, \bibinfo{institution}{Tokyo Institute of
  Technology}.
\bibitemend

\bibitemstart{kvodv}
\bibinfo{author}{J.W. Klop}, \bibinfo{author}{V.~van Oostrom} \&
  \bibinfo{author}{R.C. de~Vrijer} (\bibinfo{year}{2008}):
  \emph{\bibinfo{title}{Lambda calculus with patterns}}.
\newblock {\sl \bibinfo{journal}{Theoretical Computer Science}}
  \bibinfo{volume}{398}(\bibinfo{number}{1-3}), pp. \bibinfo{pages}{16--31}.
\bibitemend

\bibitemstart{MelliesTh}
\bibinfo{author}{Paul-Andr\'e Melli\`es} (\bibinfo{year}{1996}):
  \emph{\bibinfo{title}{Description Abstraite des Syst\`emes de
  R\'e\'ecriture}}.
\newblock \bibinfo{type}{Ph.D. thesis}, \bibinfo{school}{Universit\'e Paris
  VII}.
\bibitemend

\bibitemstart{PDR04}
\bibinfo{author}{Luca Paolini} \& \bibinfo{author}{Simona Ronchi~Della Rocca}
  (\bibinfo{year}{2004}): \emph{\bibinfo{title}{Parametric parameter passing
  Lambda-calculus}}.
\newblock {\sl \bibinfo{journal}{Information and Computation}}
  \bibinfo{volume}{189}(\bibinfo{number}{1}), pp. \bibinfo{pages}{87--106}.
\bibitemend

\bibitemstart{plotkin-std}
\bibinfo{author}{G.~Plotkin} (\bibinfo{year}{1975}):
  \emph{\bibinfo{title}{Call-by-name, call-by-value and the Lambda-calculus}}.
\newblock {\sl \bibinfo{journal}{Theoretical Computer Science}}
  \bibinfo{volume}{1}(\bibinfo{number}{2}), pp. \bibinfo{pages}{125--159}.
\bibitemend

\bibitemstart{takahashi-std}
\bibinfo{author}{M.~Takahashi} (\bibinfo{year}{1995}):
  \emph{\bibinfo{title}{Parallel reductions in lambda-calculus}}.
\newblock {\sl \bibinfo{journal}{Information and Computation}}
  \bibinfo{volume}{118}(\bibinfo{number}{1}), pp. \bibinfo{pages}{120--127}.
\bibitemend

\bibitemstart{Terese03}
\bibinfo{author}{Terese} (\bibinfo{year}{2003}): \emph{\bibinfo{title}{Term
  Rewriting Systems}}, {\sl \bibinfo{series}{Cambridge Tracts in Theoretical
  Computer Science}}~\bibinfo{volume}{55}.
\newblock \bibinfo{publisher}{Cambridge University Press}.
\bibitemend

\bibliographyend
\end{thebibliography}
\bibliographystyle{eptcs} 

\end{document}